\documentclass[orivec]{llncs}
\pdfoutput=1

\usepackage[utf8]{inputenc}
\usepackage[english]{babel}
\usepackage[babel]{csquotes}
\usepackage{microtype}

\usepackage{amssymb}
\usepackage{mathtools}
\allowdisplaybreaks

\DeclarePairedDelimiter\set{\lbrace}{\rbrace}

\DeclarePairedDelimiterX\Set[2]{\lbrace}{\rbrace}{#1\,\delimsize\vert\,\mathopen{}#2}

\newcommand{\ecc}{\bar d}
\newcommand{\p}{p}
\newcommand{\s}{s}

\usepackage{graphicx}
\usepackage{wrapfig}
\usepackage[caption=false]{subfig}
\usepackage{multirow}
\usepackage{booktabs}
\usepackage[usenames, dvipsnames,svgnames]{xcolor}

\usepackage{float}
\usepackage{enumerate}

\usepackage[hidelinks, breaklinks]{hyperref}

\title{Efficient Farthest-Point Queries in Two-Terminal~Series-Parallel Networks}
\author{Carsten Grimm\inst{1}\fnmsep\inst{2}}
\institute{Otto-von-Guericke-Universität Magdeburg, Magdeburg, Germany \and Carleton University, Ottawa, Ontario, Canada}

\begin{document}

\maketitle

\begin{abstract}
Consider the continuum of points along the edges of a network, i.e., a connected, undirected graph with positive edge weights. We measure the distance between these points in terms of the weighted shortest path distance, called the \emph{network distance}. Within this metric space, we study farthest points and farthest distances. We introduce a data structure supporting queries for the farthest distance and the farthest points on two-terminal series-parallel networks. This data structure supports farthest-point queries in \(O(k + \log n)\) time after \(O(n \log \p)\) construction time, where~\(k\) is the number of farthest points, \(n\) is the size of the network, and~\(\p\) parallel operations are required to generate the network.
\end{abstract}

\section{Introduction}

Consider a geometric network with positive edge weights. For any two points on this network (i.e., points that may be vertices or in the interior of an edge), their \emph{network distance} is the weight of a weighted shortest-path connecting them. Within this metric space, we study farthest points and farthest distances. We introduce a data structure supporting queries for the farthest distance and the farthest points on two-terminal series-parallel networks. 

As a prototype application, imagine the task to find the ideal location for a new hospital within the network formed by the streets of a city. One criterion for this optimization would be the emergency unit response time, i.e., the worst-case time an emergency crew needs to drive from the hospital to the site of an accident. However, a location might be optimal in terms of emergency unit response time, but unacceptable with respect to another criterion such as construction costs. We provide a data structure that would allow a decision maker to quickly compare various locations in terms of emergency unit response time. 

We obtain our data structure for two-terminal series-parallel networks by studying simpler networks reflecting parallel structure (parallel-path) and serial structure (bead-chains). Combining these insights, we support queries on flat series-parallel networks (abacus). Finally, we decompose series-parallel networks into a tree of nested abaci and combine their associated data structures. 

Our focus on supporting human decision makers with data structures deviates from the common one-shot optimization problems in location analysis, where we assume that only one factor determines suitable locations for some facility in a network.  Moreover, we illustrate new ways of exploiting different parallel structures of networks that may be useful for tackling related problems. 

\subsection{Preliminaries} \label{sec::preliminaries}

A \emph{network} is defined as a simple, connected, undirected graph \(N = (V,E)\) with positive edge weights.  We write \(w_{uv}\) to denote the weight of the edge \(uv \in E\) that connects the vertices \(u,v\in V\). A point \(p\) on edge \(uv\) subdivides \(uv\) into two sub-edges \(up \) and \(pv\) with \(w_{up} = \lambda w_{uv}\) and \(w_{pv} = (1-\lambda) w_{uv}\), for some \(\lambda \in [0,1]\).\footnote{Observe that \(p \notin V\) when \(\lambda \in (0,1)\) in which case none of the sub-edges \(up\) and \(pv\) are edges in \(E\). When \(\lambda = 0\) or \(\lambda = 1\), the point \(p\) coincides with \(u\) and \(v\), respectively.} 
We write \(p \in uv\) when \(p\) is on edge \(uv\) and \(p \in N\) when \(p\) is on some edge of~\(N\). The \emph{network distance} between \(p, q \in N\), denoted by \(d_N(p,q)\), is measured as the weighted length of a shortest path from \(p\) to \(q\) in \(N\). We denote the \emph{farthest distance} from \(p\) by \(\ecc_N(p)\), i.e., \(\ecc_N(p) = \max_{q \in N} d_N(p,q)\). Accordingly, we say a point \(\bar p\) on \(N\) is farthest from \(p\) if and only if \(d_N(p,\bar p) = \ecc_N(p)\). 

We develop data structures supporting the following queries in a network~\(N\). Given a point \(p\) on \(N\), what is the farthest distance from \(p\)? What are the farthest points from \(p\) in \(N\)? We refer to the former as \emph{farthest-distance query} and to the latter as \emph{farthest-point query}. The query point \(p\) is represented by the edge \(uv\) containing \(p\) together with the value \(\lambda \in [0,1]\) such that  \(w_{up} = \lambda w_{uv}\).
\begin{figure}[ht]
	\centering \vspace{-2\baselineskip}
	\subfloat[The operations.\label{fig::seriesparallel::operations}]{\includegraphics[]{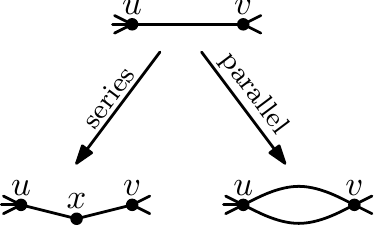}} \quad
	\subfloat[A two-terminal series-parallel network.]{\qquad\includegraphics[scale=0.8, page=1]{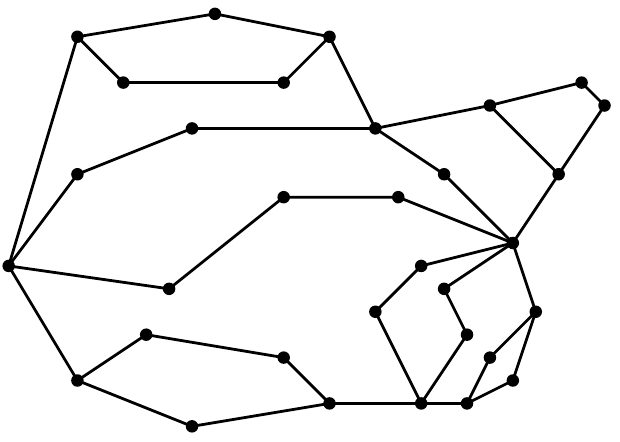}\quad}
	\caption{The operations (a) that generate two-terminal series-parallel networks (b).\label{fig::seriesparallel}} \vspace{-\baselineskip}
\end{figure}

The term \emph{series-parallel network} stems from the following two operations that are illustrated in Fig.\nobreakspace \ref {fig::seriesparallel}. The \emph{series operation}  splits an existing edge \(uv\) into two new edges \(ux\) and \(xv\) where \(x\) is a new vertex. The \emph{parallel operation} creates a copy of an existing edge. A network \(N\) is \emph{two-terminal series-parallel} when its underlying graph\footnote{The final graph is simple even if intermediate graphs have loops and multiple edges.} can be generated from a single edge \(uv\) using a sequence of series and parallel operations; the vertices \(u\) and \(v\) are called \emph{terminals} of~\(N\). We refer to the number of parallel operations required to generate \(N\)  as the \emph{parallelism} of \(N\) and to the number of series operations as the \emph{serialism} of \(N\).

A network is called \emph{series-parallel} when every bi-connected component is \emph{two-terminal series-parallel} with respect to any two vertices. In this work, we only consider bi-connected networks; in future work, we shall adapt our treatment of multiple bi-connected components from cacti~\cite{bose2015optimal} to series-parallel networks. 

\subsection{Related Work} \label{sec::relatedwork}

Duffin~\cite{duffin1965topology} studies series-parallel networks to compute the resistance of circuit boards. He characterizes three equivalent definitions of series-parallel networks and establishes their planarity. Two-terminal series-parallel networks admit linear time solutions for several problems that are NP-hard on general networks~\cite{bern1987linear,takamizawa1982linear}. Since series-parallel networks have tree-with two~\cite{brandstaedt1999graphclasses}, this applies to all problems with efficient algorithms on networks with bounded treewidth~\cite{arnborg1991easy}. 

A network Voronoi diagram subdivides a network depending on which site  is closest~\cite{hakimi1992voronoi} or farthest~\cite{erwig2000graph,okabe2008generalized} among a finite set of sites. Any data structure for farthest-point queries on a network represents a network farthest-point Voronoi diagram where all points on the network are considered sites~\cite{bose2013network}. 

A \emph{continuous absolute center} is a point on a network with minimum farthest-distance. Computing a continuous absolute center takes \(O(n)\) time on cacti~\cite{ben2007efficient} and \(O(m^2 \log n)\) time on general networks~\cite{hansen1991continuous}. As a by-product, we obtain all continuous absolute centers of a series-parallel network in \(O(n \log \p)\) time. 

\subsection{Structure and Results} \label{sec::results}

We introduce a data structure supporting queries for the farthest distance and the farthest points on two-terminal series-parallel networks. We obtain this data structure by isolating sub-structures of series-parallel networks: In Sections\nobreakspace \ref {sec::paralellfirst} and\nobreakspace  \ref {sec::beadchains}, we study networks consisting of parallel paths and networks consisting of a cycle with attached paths (bead-chains), respectively. In Section\nobreakspace \ref {sec::abacus}, we combine these results into abacus networks, which are series-parallel networks without nested structures. Finally, we combine these intermediate data structures to obtain our main result in Section\nobreakspace \ref {sec::twoterminal}. Table\nobreakspace \ref {tab::results} summarizes the characteristics of the proposed data structures and compares them to previous results.%
\begin{table}[htb] \footnotesize
\centering 
\begin{tabular}{rcccr}
 \toprule \bfseries Type & \bfseries Farthest-Point Query & \bfseries  Size   & \bfseries Construction Time & \bfseries Reference\\  \toprule
 General & \( O(k + \log n)\) &  \(O(m^2)\) & \(O(m^2 \log n)\) & \cite{bose2013network} \\ \midrule
Tree &\(O(k)\) & \(O(n)\) & \(O(n)\) & \cite{bose2015optimal}\\
Cycle & \(O(\log n)\) & \(O(n)\) & \(O(n)\) & \cite{bose2015optimal} \\
Uni-Cyclic & \(O(k+\log n)\) & \(O(n)\) & \(O(n)\) & \cite{bose2015optimal}\\
Cactus &  \(O(k+\log n)\) & \(O(n)\) & \(O(n)\) & \cite{bose2015optimal} \\ \midrule
Parallel-Path&  \(O(k + \log n )\) & \(O(n)\) & \(O(n)\) & this work \\
Bead-Chain  & \(O(k + \log n )\) & \(O(n)\) & \(O(n)\) & this work\\
Abacus  &  \( O(k + \log n)\) & \(O(n)\) & \(O(n \log \p)\)& this work\\ 
Series-Parallel  & \(O(k + \log n )\)& \(O(n )\) & \(O(n \log \p)\)& this work \\\bottomrule \\
\end{tabular}
\caption{The traits of our data structures for queries in different types of networks, with \(n\) vertices, \(m\) edges, \(k\) reported farthest points, and parallelism \(\p\).} \label{tab::results}  \vspace{-\baselineskip}
\end{table}

%

\section{Parallel-Path Networks} \label{sec::paralellfirst}

\begin{wrapfigure}[10]{r}{6cm}
	\centering \vspace{-1.5\baselineskip}
	\includegraphics[scale=0.9]{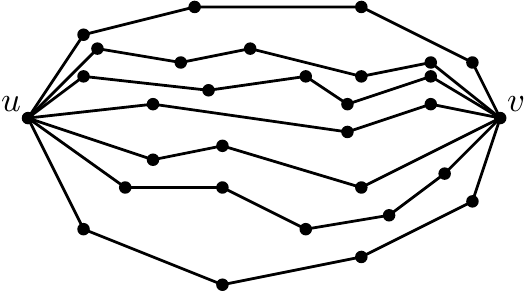} \vspace{-0.5\baselineskip}
	\caption{A parallel-path two-terminal series-parallel network with parallelism \(\p = 7\).\label{fig::parallelfirst}}
\end{wrapfigure}

A \emph{parallel-path  network} consists of a bundle of edge disjoint paths connecting two vertices \(u\) and \(v\), as illustrated in Fig.\nobreakspace \ref {fig::parallelfirst}. In terms of series-parallel networks, parallel-path networks are generated from an edge \(uv\) using a sequence of parallel operations followed by a sequence of series operations.

Let \(P_1, P_2, \dots, P_\p\) be the paths of weighted lengths \(w_1 \le w_2 \le \dots \le w_\p\) between the terminals \(u\) and \(v\) in a parallel-path  network \(N\). Consider a shortest path tree\footnote{More precisely, we consider \emph{extended shortest path trees}~\cite{okabe2008generalized} which result from splitting each non-tree edge  \(st\) of a shortest path tree into two sub-edges \(sx\) and \(xt\), where all points on \(sx\) reach the root through \(s\) and all points on \(xt\) reach the root through \(t\).} from a query point \(q \in N\).  As depicted in Fig.\nobreakspace \ref {fig::parallelfirst::threecases}, there are three cases: either all shortest paths from \(q\) reach \(v\) via \(u\) (\emph{left case}), or all shortest paths from \(q\) reach \(u\) via \(v\) (\emph{right case}), or neither (\emph{middle case}). We distinguish the three cases using the following notation. Let  \(\bar x_i\) denote the farthest point from \(x \in \set{u,v}\) among the points of path \(P_i\), i.e., \(\bar x_i\) is a point on \(P_i\) such that \(d(x,\bar x_i) = \max_{y \in P_i} d(x,y)\).  Together with Fig.\nobreakspace \ref {fig::parallelfirst::threecases}, the next lemma justifies our choice of the names \emph{left case}, \emph{middle case}, and \emph{right case}.%
\begin{figure}[H]
	\centering \vspace{-2\baselineskip}
	\subfloat[][The left case. \label{fig::parallelfirst::threecases::leftcase}]{\includegraphics[page=2]{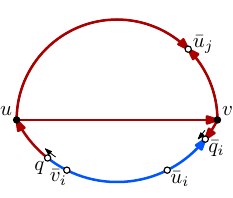}}\hfill
	\subfloat[][The middle case. \label{fig::parallelfirst::threecases::middlecase}]{\includegraphics[page=3]{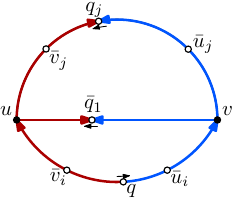}}\hfill
	\subfloat[][The right case. \label{fig::parallelfirst::threecases::rightcase}]{\includegraphics[page=4]{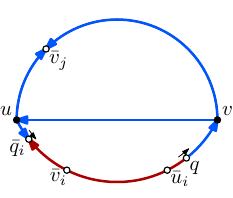}}
	
	\caption{The three cases for queries in parallel-path networks. Consider the shortest path tree from a query point \(q\in P_i\) along the paths \(P_1\) (center), \(P_i\) (bottom),~and~\(P_j\)~(top).
	In the left case~\protect\subref{fig::parallelfirst::threecases::leftcase}, we reach \(v\) via \(u\). In the right case~\protect\subref{fig::parallelfirst::threecases::rightcase}, we reach \(u\) via \(v\). In the middle case~\protect\subref{fig::parallelfirst::threecases::middlecase}, neither holds, i.e., we \emph{enter} path \(P_1\) from both terminals \(u\) and \(v\). Points colored red are reached fastest via a path through \(u\) or towards \(u\) along \(uv\), while points colored blue are reached fastest via a path through \(v\) or towards \(v\) along \(uv\).
	\label{fig::parallelfirst::threecases}
	} \vspace{-1\baselineskip}
\end{figure}

\begin{lemma} \label{thm::parallelfirst::threecases} Consider a query \(q \) from  the \(i\)-th path of a parallel-path network.%
\begin{enumerate}[(i)]
\item We are in the left case when \(q\) lies on the sub-path from \(u\) to \(\bar v_i\) with \(q \ne \bar v_i\),
\item we are in the middle case when \(q\) lies on the sub-path from \(\bar v_i\) to \(\bar u_i\), and
\item we are in the right case when \(q\) lies on the sub-path from \(\bar u_i\) to \(v\) with \(q \ne \bar u_i\).
\end{enumerate}
\end{lemma}
\begin{proof}
Assume we are in the left case, where every shortest path from the query point \(q \in P_i\) to \(v\) contains \(u\), i.e., \(d(q,v) = d(q,u) + d(u,v)\) and \(d(q, v) < w_{qv}\). The latter implies \(q\ne \bar v_i\), since \(d(v, \bar v_i) = w_{\bar v_iv}\). Moreover, \(\bar v_i\) cannot lie on the sub-path from \(q\) to \(u\) along \(P_i\), since otherwise 
\[
d(q,v) = d(q,\bar v_i) + d(\bar v_i, v) \stackrel{q \ne \bar v_i}> d(\bar v_i, v) \enspace,  
\] 
contradicting the choice of \(\bar v_i\) as farthest point from \(v\) on \(P_i\). Therefore, \(q\) must lie between \(u\) and \(\bar v_i\) along \(P_i\) when we are in the left case.

Conversely, assume \(q\) lies between \(u\) and \(\bar v_i\) along \(P_i\) with \(q \ne \bar v_i\). No shortest path from \(q\) to \(v\) can contain \(\bar v_i\) in its interior. Hence, every shortest path from \(q\) to \(v\) reaches \(v\) via \(u\), i.e., the left case applies. Symmetrically, the right case applies if and only if \(q \ne u_i\) lies between \(\bar u_i\) and \(v\).  Consequently, only the middle case remains for all points on the sub-path from \(\bar v_i\) to \(\bar u_i\) and the claim follows. \qed
\end{proof}
Using Lemma\nobreakspace \ref {thm::parallelfirst::threecases}, we deal with the three cases as follows.

\subsubsection{Left Case and Right Case}

In the left case, every shortest path from \(q \in P_i\) to any point outside of \(P_i\) leaves \(P_i\) through \(u\). Hence, the farthest point from \(q\) on \(P_j\)  with \(j \ne i\) is the farthest point \(\bar u_j\) from \(u\) on \(P_j\). The distance from \(q\) to \(\bar q_j\) is \(d(q,\bar q_j) = d(q,u) + d(u,\bar u_j) = w_{qu} + \frac{w_1 + w_j}{2}\). 
On the other hand, the farthest point \(\bar q_i\) 
from \(q\) on \(P_i\) itself moves from \(\bar u_i\) to \(v\) as \(q\) moves from \(u\) to \(\bar v_i\) maintaining  a distance of \(d(q, \bar q_i) = \frac{w_1 + w_i}{2}\)
. Therefore, the farthest distance from \(q\) in \(N\) is 
\begin{align*}
	\ecc(q) &= \max\left[ \frac{w_1 + w_i}{2},\,  \max_{j \ne i} \left(w_{qu} + \frac{w_1 + w_j}{2} \right) \right] \\
	&=  \begin{dcases}  w_{qu} + \frac{w_1 + w_\p}{2}     & \text{if } i \ne \p \\  
					w_{qu} + \frac{w_1 + w_{\p-1}}{2}    & \text{if } i = \p \text{ and } \frac{w_\p - w_{\p-1}}{2} \le  w_{qu} \\
					\frac{w_1 + w_\p}{2}                 & \text{if } i = \p \text{ and } \frac{w_\p - w_{\p-1}}{2} \ge  w_{qu}
	   \end{dcases} \enspace .
\end{align*}
The first case means that, for queries from anywhere other than \(P_\p\), the farthest points lie on the \(u\)-\(v\)-paths of maximum length. The second and third case distinguish whether \(P_\p\) contains a farthest point for queries from \(P_\p\) itself. Accordingly, we answer a farthest point query from \(q \in P_i\) in the left case as follows.
\begin{itemize}
\item If \(i \ne \p\), we report all \(\bar u_j\) where \(w_j = w_\p\) and \(i \ne j\). \item If \(i = \p\) and \(\frac{w_\p - w_{\p-1}}{2} \le w_{qu}\), we report all  \(\bar u_j\) where \(w_j = w_{\p-1}\) and~\(j \ne \p\). \item If \(i = \p\) and \(\frac{w_\p - w_{\p-1}}{2} \ge  w_{qu}\), we report the farthest point \(\bar q_\p\) from \(q\) on \(P_\p\) using a binary search along the sub-path of \(P_\p\) from \(\bar u_\p\) to \(v\). \end{itemize}   The overlap between the last two cases covers the boundary case when a query from \(P_\p\) yields a farthest point on  \(P_\p\) itself and farthest points on other \(u\)-\(v\)-paths.

Swapping \(u\) and \(v\) above yields the procedure for the right case. Thus, answering farthest-point queries takes \(O(k+\log n)\) time in the left and right cases.

\subsubsection{Middle Case}

In the middle case, there are no farthest points from \(q\) on \(P_i\) itself and every path \(P_j\) with \(j \ne i\) contains points that we reach from \(q\) via \(u\) as well as points that we reach from \(q\) via \(v\). Let \(\bar q_j\) be the farthest point from \(q\) along the cycle formed by \(P_j\) and \(P_i\).  Since the distance from \(q\) to \(\bar q_j\) is \(d(q,\bar q_j) = \frac{w_i + w_j}{2}\), the farthest distance \(\ecc(q)\) from \(q\) in \(N\) is 
\[
	\ecc(q) = \max_{j \ne i} \left( \frac{w_i + w_j}{2} \right)
	=	\begin{dcases} 
			\frac{w_i 		+ w_\p}{2} 		& \text{if } i \ne \p \\
			\frac{w_\p 	+ w_{\p-1}}{2} 	& \text{if } i = \p  
		\end{dcases} \enspace .
\]

The first case applies for queries from anywhere other than \(P_\p\) who have their farthest points on the longest \(u\)-\(v\)-paths, i.e., on the paths \(P_j\) with \(w_j = w_\p\). The~second~case applies for queries from \(P_\p\) who have their farthest points on the second longest \(u\)-\(v\)-paths, i.e., on the paths \(P_j\) with \(w_j = w_{\p -1}\). Using binary search, we can answer a farthest point query from \(q \in P_i\) in the middle case by reporting the points~\(\bar q_j\) on those \(k\) paths \(P_j\) that contain farthest points from~\(q\).  To improve the resulting query time of \(O(k \log n)\), we take a closer look at the position of \(\bar q_j\) relative to \(\bar u_j\) and \(\bar v_j\). As illustrated in Fig.\nobreakspace \ref {fig::thm::parallelwalk}, the farthest point \(\bar q_j\) from \(q \in P_i\) along \(P_j\) moves from \(\bar u_j\) to \(\bar v_j\) as \(q\) moves from \(\bar v_i\) to \(\bar u_i\).

\begin{lemma} \label{thm::parallelwalk}
Let \(N\) be a parallel-path network with terminals \(u\) and \(v\). For any point \(x\in N\), let \(\bar x_i\) denote the farthest point from \(x\) along the \(u\)-\(v\)-path \(P_i\).
\begin{enumerate}[(i)]
\item The sub-path from \(\bar v_i\) to \(\bar u_i\) has length \(d(u,v)\).
\item For every point \(q\) along the sub-path from \(\bar v_i\) to \(\bar u_i\), the sub-path from \(\bar v_i\) to~\(q\) has the same length as the sub-path from \(\bar u_j\) to \(\bar q_j\) for any \(j\ne i\).
\end{enumerate} 
\end{lemma}
\begin{proof} We have \(w_{u\bar v_i} = w_{v\bar u_i}\) for every path \(P_i\), since
\[
	w_{u\bar v_i} + w_{\bar v_i\bar u_i} = w_{u\bar u_i}  = \frac{d(u,v) + w_{i}}{2} = w_{v \bar v_i} = w_{v\bar u_i} + w_{\bar u_i \bar v_i} \enspace .
\]

The above implies the first claim, since adding a nutritious zero yields 
\[
w_{\bar v_i \bar u_i} = w_{\bar v_i \bar u_i} + (w_{u\bar v_i} - w_{v\bar u_i}) = d(u, \bar u_i) - w_{v\bar u_i} = d(u,v) + w_{v\bar u_i} - w_{v\bar u_i} = d(u,v)  \enspace .
\]

For the second claim, we compare the length of the shortest path from \(q\) to \(\bar q_j\) via \(v\) with the one via \(u\).  By studying Fig.\nobreakspace \ref {fig::thm::parallelwalk}, we obtain the following.%
\begin{align}
	d(q,\bar q_j) &= w_{q \bar v_i} + w_{\bar v_i u} + w_{u\bar v_j}  + w_{\bar v_j\bar u_j} - w_{\bar q_j \bar u_j} \label{eq::parallelwalk::1} \\
	d(q,\bar q_j) &= w_{\bar q_j \bar u_j} + w_{v\bar u_j}  + w_{v\bar u_i} + w_{\bar v_i\bar u_i} - w_{q \bar v_i} \label{eq::parallelwalk::2} 
\end{align}
We generate nutritious zeros from the identities \(w_{u \bar v_i} = w_{v \bar u_i}\), \(w_{u \bar v_j} = w_{v \bar u_j}\), and \(w_{\bar u_i \bar v_i} = d(u,v) = w_{\bar u_j\bar v_j}\). Rearranging terms yields the second claim, as
\begin{align*} 
 2 w_{q \bar v_i} &= 2 w_{q \bar v_i} + \overbrace{w_{\bar v_i u} - w_{v\bar u_i}}^{=0} +  \overbrace{w_{u\bar v_j} - w_{v\bar u_j} }^{=0} + \overbrace{w_{\bar v_j\bar u_j} -w_{\bar v_i\bar u_i} }^{=0} + 2 (\overbrace{ w_{\bar q_j \bar u_j} - w_{\bar q_j \bar u_j} }^{=0}) \\
  &\stackrel{\eqref{eq::parallelwalk::1}}= d(q,\bar q_j)  - \left[  w_{\bar q_j \bar u_j} + w_{v\bar u_j}  + w_{v\bar u_i} + w_{\bar v_i\bar u_i} - w_{q \bar v_i} \right]  + 2 w_{\bar q_j \bar u_j} \\ 
  &\stackrel{\eqref{eq::parallelwalk::2}}= d(q,\bar q_j)  - d(q,\bar q_j) +  2 w_{\bar q_j \bar u_j}  =  2 w_{\bar q_j \bar u_j} \enspace . \tag*{\qed}
\end{align*}
\end{proof}%

\begin{wrapfigure}[12]{l}{5cm}
	\centering 
	\includegraphics[page=5]{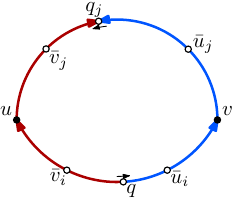} \vspace{-0.5\baselineskip}
	\caption{The positions of the points along the cycle \(P_i \cup P_j\) in Lemma\nobreakspace \ref {thm::parallelwalk}. \label{fig::thm::parallelwalk}}
\end{wrapfigure}
Using Lemma\nobreakspace \ref {thm::parallelwalk}, we interpret the searches for \(\bar q_j\) on the sub-path from \(\bar v_j\) to~\(\bar u_j\), as a single search \emph{with a common key} \(\bar q\) in \emph{multiple lists} (the \(\bar v_i\)-\(\bar u_i\)-sub-paths) of \emph{comparable search keys} (the vertices along these sub-paths). Using \(O(n)\) time, we construct a fractional cascading data structure~\cite{chazelle1986fractional} supporting predecessor queries on the sub-paths from \(\bar v_j\) to~\(\bar u_j\) for those paths \(P_j\) where \(w_j = w_{\p-1}\).\footnote{We consider paths of length \(w_{p-1}\) instead of \(w_p\), because we treat \(P_\p\) separately.} 

We answer a farthest-point query from \(q \in P_i\) as follows. If \(i \ne \p\), we locate and report \(\bar q_\p\) along \(P_\p\) in \(O(\log n)\) time. If \(i = \p\) or \(w_\p = w_{\p-1}\), the remaining farthest points from \(q\) are the \(\bar q_j\) where \(j \ne i\) and \(w_j = w_{\p-1}\); we report them in \(O(k + \log n)\) time using the fractional cascading data structure. This query might report a point on \(P_i\), which would be \(\bar q_i\) for queries from outside~\(P_i\). For queries from within \(P_i\), we omit this artifact.


\begin{theorem} \label{thm::parallelfirst::twoterminal}
Let \(N\) be a parallel-path network with \(n\) vertices. There is a data structure with \(O(n)\) size and  \(O(n)\) construction time supporting \(O(k + \log n)\)-time farthest-point queries on \(N\), where \(k\) is the number of farthest points. 
\end{theorem}


\section{Bead-Chain Networks} \label{sec::beadchains}

\begin{wrapfigure}[14]{r}{5cm}
	\centering \vspace{-0.5\baselineskip}
	\includegraphics{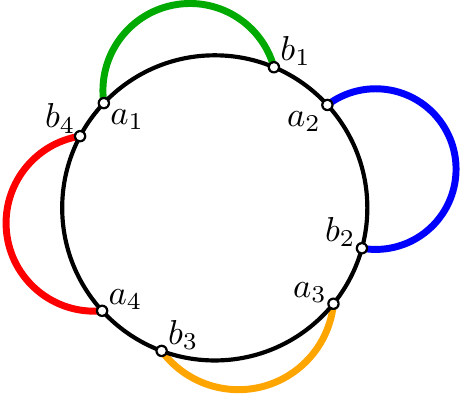} \vspace{-0.5\baselineskip}
	\caption{A bead-chain with four arcs (colored) around its cycle (black).\label{fig::beadchain}}
\end{wrapfigure}
A \emph{bead-chain network} consists of a main cycle with attached arcs so that each arc returns to the cycle before the next one begins. An example is depicted in Fig.\nobreakspace \ref {fig::beadchain}. Bead-chains are series-parallel networks where we first subdivide a cycle using series operations, then we apply at most one parallel operation to each edge of this cycle followed by series operations that further subdivide the arcs and cycle.  

Consider a bead-chain network \(N\) with main cycle \(C\) and arcs \(\alpha_1,\dots, \alpha_\s\). Let \(a_i\) and \(b_i\) be the vertices connecting \(C\) with the \(i\)-th arc. Without loss of generality, the path \(\beta_i\) from \(a_i\) to \(b_i\) along \(C\) is at most as long as \(\alpha_i\). Otherwise, we swap the roles of \(\alpha_i\) and  \(\beta_i\). 

We first study the shape of the function \(\hat d_i(x)\) that describes the farthest distance from points along the main cycle to any point on the \(i\)-th arc, i.e., 
\(
	\hat d_i(x) = \max_{y \in \alpha_i} d(x,y)
\). When considering only the \(i\)-th arc, we have a parallel-path network with three paths. Let \(\bar x\)  denote the farthest point from \(x \in C\) on \(C\) itself and let \(\hat x_i\) denote the farthest point from \(x\) on arc \(\alpha_i\). From the analysis in the previous section, we know that \(\bar d_i(x)\) has the shape depicted in Fig.\nobreakspace \ref {fig::arcdistance}. 

\begin{figure}[htb]
	\centering
	\vspace{-1\baselineskip}
	\subfloat[][The main cycle with \(\alpha_i\) (blue).]{\includegraphics[page=2]{beadchain-crop}\quad} 
	\subfloat[][The distance to \(\hat x_i\).]{\includegraphics[]{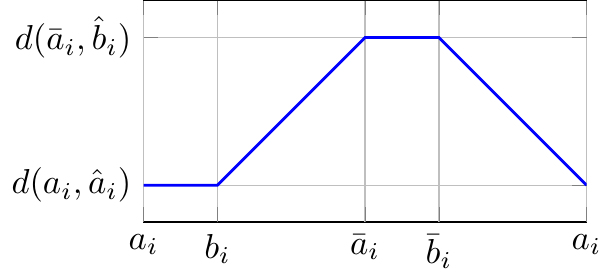} }
	
	\caption{The shape of the function \(\hat d_i(x)\) describing the distance from \(x\in C\) to the farthest point \(\hat x_i\) from \(x\) among the points on  the \(i\)-th arc \(\alpha_i\).\label{fig::arcdistance}} \vspace{-1\baselineskip}
\end{figure}

When walking along the main cycle, we encounter \(a_i\), \(b_i\), \(\bar a_i\), and \(\bar b_i\) in this order or its reverse. From \(a_i\) to \(b_i\), the point \(\hat x_i\) moves from \(\hat a_i\) to \(\hat b_i\) maintaining a constant distance. From \(b_i\) to \(\bar a_i\), the point \(\hat x_i\) stays at \(\hat b_i\) increasing in distance. From \(\bar a_i\) to \(\bar b_i\), the point \(\hat x_i\) moves from \(\hat b_i\) back to \(\hat a_i\), again, at a constant distance. Finally, \(\hat x_i\) stays at \(\hat a_i\) with decreasing distance when \(x\) moves from \(\bar b_i\) to \(a_i\). 

Since the farthest distance changes at the same rate when we move towards or away from the current farthest-point, the increasing and decreasing segments of any two functions \(\hat d_i\) and \(\hat d_j\) have the same slope except for their sign.

The height of the upper envelope \(\hat D\) of the functions \(\hat d_1,\dots, \hat d_\s\) at \(q \in C\) indicates the farthest distance from \(q\) to any point on the arcs and the \(i\)-th arc contains a farthest point from \(q\) when \(\hat d_i\) coincides with \(\hat D\) at \(q\). We construct \(\hat D\) in linear time using the shape of the functions \(\hat d_1,\dots, \hat d_\s\) described above. 

We need to consider an arc separately from the other arcs when it is \emph{too long}. We call an arc \(\alpha_i\) \emph{overlong} when the path \(\beta_i\) is longer then remainder \(\gamma_i\) of the main cycle. Figure\nobreakspace \ref {fig::overlong} illustrates an overlong arc \(\alpha_i\) with its function \(\hat d_i\).

\begin{figure}[htb]
	\centering%
	\subfloat[][An overlong arc (blue).]{\includegraphics[page=3]{beadchain-crop}} \quad
	\subfloat[][The farthest distance to an overlong arc.]{\includegraphics[]{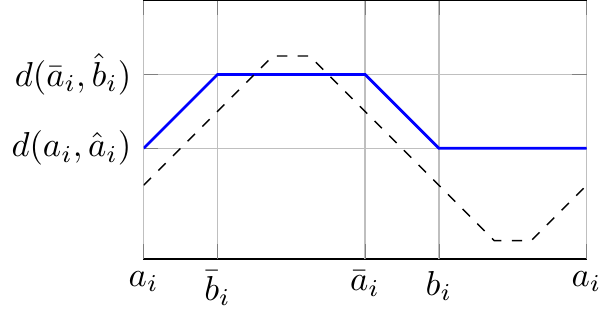} }
	
	\caption{An overlong arc \(\alpha_i\) (blue) in a bead-chain network where \(\beta_i\) (green) is longer then the remaining cycle \(\gamma_i\) (orange). The shape of \(\hat d_i\) is the same as for non-overlong arcs, but its high plateau may horizontally overlap with the high plateau of other arcs. \label{fig::overlong}}
\end{figure}

\begin{lemma} \label{thm::overlong} Every bead-chain network \(N\) has at most one overlong arc and the functions \(\hat d_1,\dots,\hat d_\s\) of the remaining arcs \(\alpha_1,\dots, \alpha_\s\) satisfy the following.%
\begin{enumerate}[(i)]
\item The high plateaus of \(\hat d_1,\dots,\hat d_\s\) appear in the order as their arcs \(\alpha_1,\dots, \alpha_\s\) appear along the cycle and no two high plateaus overlap horizontally.

\item The low plateaus of \(\hat d_1,\dots,\hat d_\s\) appear in the order as their arcs \(\alpha_1,\dots, \alpha_\s\) appear along the cycle and no two low plateaus overlap horizontally.
\end{enumerate}
\end{lemma}
\begin{proof}
The existence of two overlong arcs \(\alpha_i\) and \(\alpha_k\) leads to the contradiction \(
	w_{\beta_i} > w_{\gamma_i} > w_{\beta_k} > w_{\gamma_k} > w_{\beta_i} 
\), since \(\gamma_i\) contains \(\beta_k\) and \(\gamma_k\) contains \(\beta_i\). 

Let \(\alpha_1,\dots, \alpha_\s\) be the non-overlong arcs of~\(N\) as they appear along the cycle. For every arc \(\alpha_i\), with \(i =1,\dots,\s\), the farthest points \(\bar a_i\) and \(\bar b_i\) from the endpoints \(a_i\) and \(b_i\) of \(\alpha_i\) appear along \(\gamma_i\).  Therefore, the points \(\bar a_1,\bar b_1,\bar a_2, \bar b_2,\dots, \bar a_\s\), and \(\bar b_\s\) appear in this order along the cycle \(C\). Claim (i) follows, since the high plateau of \(\hat d_i\) lies between \(\bar a_i\) and \(\bar b_i\). Claim (ii) follows, since any non-overlong arc \(\alpha_i\) has its low plateau along \(\beta_i\) and since \(\beta_1,\dots,\beta_\s\) are separate by definition. \qed
\end{proof}

As suggested by Lemma\nobreakspace \ref {thm::overlong}, we incrementally construct the upper envelope of the functions \(\hat d_i\) corresponding to non-overlong arcs and treat a potential overlong arc separately. When performing a farthest-point query from the cycle, we first determine the farthest distance to the overlong arc and the farthest distance to all other arcs. Depending on the answer we report farthest points accordingly. 

\begin{lemma} \label{thm::envelope} Let \(\alpha_1,\dots,\alpha_s\) be the arcs of a bead-chain that has no overlong arc. Computing the upper envelope \(\hat D\) of \(\hat d_1,\dots, \hat d_\s\) takes \(O(\s)\) time.
\end{lemma}
\begin{proof} We proceed in two passes: in the first pass, we consider only the high plateaus and non-constant segments of \(\hat d_1,\dots, \hat d_\s\); the respective low plateaus are replaced by extending the corresponding non-constant segments. In the second pass, we traverse the partial upper envelope from the first pass again and compare it with the previously omitted low plateaus, thereby constructing \(\hat D\). 

\begin{figure}[hbt]
	\centering
	\includegraphics[scale=0.9]{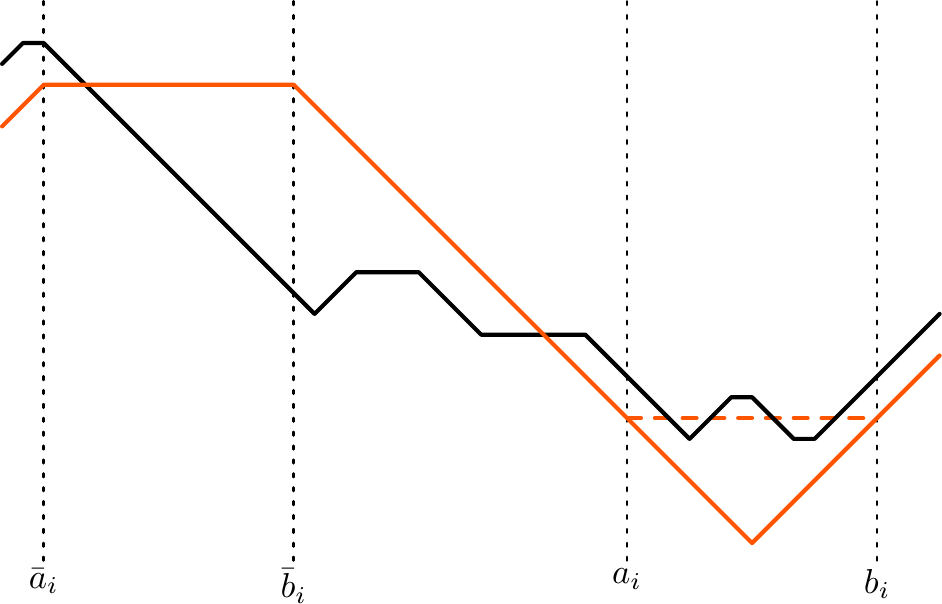}
	\caption{An incremental step where we construct \(\hat D_i'\) from \(\hat D_{i-1}'\) (black) and \(\hat d_i'\) (orange). The treatment of the low plateau of \(\hat d_i\) (dashed) is deferred to the second pass. \label{fig::firstpass}} 
\end{figure}

Let \(\hat d_i'\) be the function resulting from replacing the low plateau of \(\hat d_i\) by extending its non-constant segments and let \(\hat D_{i}'\) be the upper envelope of \(\hat d_1',\dots, \hat d_i'\). In the first pass, we construct \(\hat D_\s'\) incrementally. Assume we have \(\hat D_{i-1}'\) and would like to obtain \(\hat D_i'\) by inserting \(\hat d_i'\) into \(\hat D_{i-1}'\), as illustrated in Fig.\nobreakspace \ref {fig::firstpass}. We perform this insertion by walking from \(\bar a_i\)---the left endpoint of the high plateau of \(\hat d_i\)---in both directions updating the current upper envelope \(\hat D_{i-1}\). Locating \(\bar a_i\) takes constant time, since \(\bar a_i\) is the first bending point to the right of \(\bar b_{i-1}\).

There is no more than one increasing segment of \(\hat D_{i-1}'\) between \(\bar a_i\) and \(\bar b_i\). Assume, for a contradiction that there are two increasing segments \(s_1\) and \(s_2\) between \(\bar a_i\) and \(\bar b_i\). Neither of them has their higher endpoint between \(\bar a_i\) and~\(\bar b_i\), since there would be two horizontally overlapping high plateaus, otherwise. Since \(s_1\) and \(s_2\) have the same slope, only one of them can be part of the upper envelope.\footnote{When the segments \(s_1\) and \(s_2\) happen to overlap, we consider only the segment of the arc distance function that was inserted first to be part of the upper envelope.} Analogously, the same holds for decreasing segments. Therefore, inserting the high plateau of \(\hat d_i\) into the upper envelope \(\hat D_{i-1}'\) takes constant time. 

If the decreasing segment of \(\hat d_i'\) appears along \(\hat D_i'\) at all, then it appears at~\(\bar b_i\). We update the previous upper envelope \(\hat D_{i-1}'\) by walking from \(\bar b_i\) towards \(a_i\) until the decreasing segment of \(\hat d_i'\) vanishes below \(\hat D_{i-1}'\). We charge the costs for this walk to the segments that are removed from the previous upper envelope \(\hat D_{i-1}'\). We proceed in the same fashion with the increasing segment of \(\hat d_i'\) by walking from \(\bar a_i\) towards \(b_i\). Each non-constant segment appears at most once along any intermediate upper envelope and is never considered again after its removal. Therefore, the total cost for inserting all non-constant segments---and, thus, the total cost of constructing \(\hat D_\s'\)---amounts to \(O(\s)\).

In the second pass, we construct the desired upper envelope \(\hat D\) from \(\hat D_\s'\). Since no two low plateaus overlap horizontally, we simply walk along \(\hat D_\s'\) comparing its height to the height of the current low plateau, if any. This takes \(O(s)\) time, since \(\hat D_\s'\) has \(O(\s)\) bending points and since there are \(\s\) low plateaus.  \qed
\end{proof}

To answer a farthest-point query from \(q\in C\), we need to find its farthest arcs, i.e., the arcs containing farthest points from \(q\). Suppose each point along the main cycle~\(C\) has exactly one farthest arc. Then we could subdivide \(C\) depending on which arc is farthest and answer a farthest-arc query by identifying the function among \(\hat d_1,\dots, \hat d_\s\) that defines the upper envelope \(\hat D\) on the sub-edge containing~\(q\). On the other hand, there could be multiple farthest arcs when several functions among \(\hat d_1,\dots, \hat d_\s\) have overlapping increasing or decreasing segments. In~this~case, we could store the at most two farthest arcs from plateaus directly with the corresponding segments of~\(\hat D\). However, storing the farthest arcs from increasing and decreasing segments directly would lead to a quadratic construction time. Instead, we rely on the following observation. An arc \(\alpha\) is considered \emph{relevant} when there exists some point \(x \in C\) such that \(\alpha\) is a farthest arc for \(x\) and \(\alpha\) is considered \emph{irrelevant} when there is no such point on the main cycle.

\begin{lemma}  \label{thm::aux::nogap}
Let \(\alpha_i\), \(\alpha_j\), and \(\alpha_k\) be arcs that appear in this order in a bead-chain without overlong arc. The arc \(\alpha_j\) is irrelevant when \(\alpha_i\) and \(\alpha_k\) are farthest arcs from some query point \(q\) such that \(\hat d_i\) and \(\hat d_k\) are both decreasing/increasing at \(q\).
\end{lemma}
\begin{proof}
As illustrated in Fig.\nobreakspace \ref {fig::nogaps}, \(q\) lies between \(a_i\) and \(\bar b_k\), since \(\hat d_i\) and \(\hat d_k\) are both decreasing at \(q\). Hence, \(q\) lies between \(a_j\) and \(\bar b_j\), i.e., \(\hat d_j\) is decreasing at \(q\), as~well. We show \(\hat d_j(x) < \hat d_i(x)\) for all \(x\in C\),  which implies that \(\alpha_j\) is irrelevant.%
\begin{figure}[htb]
	\centering%
	\vspace{-1\baselineskip}
	\subfloat[][The three arcs.]{\includegraphics[scale=1.2, page=4]{beadchain-crop}} \quad 
	\subfloat[][The comparison of \(\hat d_i\) (blue) and \(\hat d_j\) (green).]{\includegraphics[]{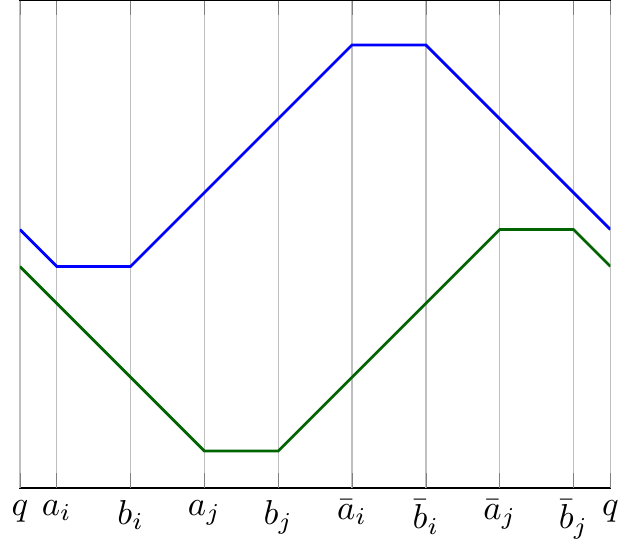} }
	
	\caption{The constellation from Lemma\nobreakspace \ref {thm::aux::nogap}, where \(\alpha_i\) and \(\alpha_k\) are farthest arcs from \(q\) where the farthest distances decreases as \(q\) moves in clockwise direction~(a). A comparison of the arc distance functions \(\hat d_i\) and \(\hat d_j\) reveals that \(\alpha_j\) is irrelevant in this case~(b).   \label{fig::nogaps}} \vspace{-1\baselineskip}
\end{figure}

Consider the difference \(\Delta(x) \coloneqq \hat d_i(x) - \hat d_j(x)\). We have \(\Delta(a_i) = \Delta(q)>0\), as \(\hat d_i\) and \(\hat d_j\) are decreasing  with the same slope from \(q\) to \(a_i\). We observe that \(\Delta(b_i) = \Delta(a_i) + d(a_i,b_i)\), as \(\hat d_i\) remains constant from \(a_i\) to \(b_i\) while \(\hat d_j\) decreases. We have \(\Delta(a_j) = \Delta(b_i) + 2 d(b_i,a_j)\), as  \(\hat d_i\) increases from \(b_i\) to \(a_j\) while \(\hat d_j\) decreases. We continue in this fashion and obtain the following description of \(\Delta\). 
\begin{align*}
	\Delta(a_i) &= \Delta(q) 				& \Delta(\bar a_i) &= \Delta(b_j)  \\
	\Delta(b_i) &= \Delta(a_i) + d(a_i,b_i)		& \Delta(\bar b_i) &= \Delta(\bar a_i) - d(\bar a_i, \bar b_i)\\
	\Delta(a_j) &= \Delta(b_i) + 2d(b_i,a_j)	& \Delta(\bar a_j) &= \Delta(\bar b_i) - 2 d(\bar b_i, \bar a_j)\\ 
	\Delta(b_j) &= \Delta(a_j) + d(a_j,b_j)	& \Delta(\bar b_j) &= \Delta(\bar a_j) - d(\bar a_j, \bar b_j)  
\end{align*}
The claim follows as the above implies \(\Delta(x) \ge \Delta(q) > 0\) for all \(x \in C\), due to the symmetries \(d(a_i,b_i) = d(\bar a_i,\bar b_i)\), \(d(b_i,a_j) = d(\bar b_i,\bar a_j)\), and \(d(a_j,b_j) = d(\bar a_j, \bar b_j)\). \qed
\end{proof}

\begin{corollary}  \label{thm::nogaps}
Let \(q\) be a point on the cycle of a bead-chain with no overlong arc. The farthest arcs from \(q\) that correspond to decreasing/increasing segments of~\(\hat D\) form one consecutive sub-list of the circular list of relevant arcs. 
\end{corollary}

Using Corollary\nobreakspace \ref {thm::nogaps}, we answer farthest-arc queries from the main cycle of a bead-chain network without overlong arc as follows. When \(\hat D\) has a plateau at the query point \(q\), we report the at most two farthest arcs stored with this plateau. When \(\hat D\) has an increasing/decreasing segment at \(q\), we first report the farthest arc \(\alpha\) that is stored directly with this segment. We report the remaining farthest arcs by cycling through the circular list of relevant arcs starting from \(\alpha\) in both directions until we reach a relevant arc that is no longer farthest from \(q\).

\begin{theorem} \label{thm::beadchains}
Let \(N\) be a bead-chain network with \(n\) vertices. There is a data structure with \(O(n)\) size and  \(O(n)\) construction time supporting \(O(k + \log n)\)-time farthest-point queries on \(N\), where \(k\) is the number of farthest points. 
\end{theorem}
\begin{proof}
Let \(N\) be a bead-chain network with main cycle \(C\) and arcs \(\alpha, \alpha_1,\dots, \alpha_\s\) where only \(\alpha\) may be overlong. Let \(N \setminus \alpha\) be the network obtained by the removing the potentially overlong arc from \(N\), i.e., \(N \setminus \alpha := C \cup \alpha_1 \cup \dots \cup \alpha_s\). 

We support queries from the main cycle \(C\) by constructing (i) a data structure for queries in the parallel-path network \(\alpha \cup C\) and (ii) a data structure for queries in the bead-chain network \(N \setminus \alpha\) that has no overlong arc. This construction takes 
linear time, due to Theorem\nobreakspace \ref {thm::parallelfirst::twoterminal} and\nobreakspace Lemma\nobreakspace \ref {thm::envelope}. For~a~farthest-point query from \(q \in C\), we perform a farthest-distance query from \(q\) in  \(\alpha \cup C\) and in~\(N \setminus \alpha\). Depending on which of these two queries reported the largest farthest distance, we conduct the appropriate farthest-point queries in \(\alpha \cup C\) or in \(N \setminus \alpha\).

We support queries from arc \(\alpha_i\) (from \(\alpha\)) by combining (i) a data structure for farthest-point queries in the cycle \(\alpha_i \cup \beta_i\) (in \(\alpha \cup C\)), with (ii) the data structure supporting farthest-point queries in \(N\) from the cycle \(C\) described above, and (iii) the partial upper envelope \(\hat D'\) from the proof of Lemma\nobreakspace \ref {thm::envelope}.\footnote{Recall that \(\hat D'\) was the upper envelope of the functions that result from replacing the low plateaus of \(\hat d_1,\dots,\hat d_\s\) by extending their increasing/decreasing segments.}

We begin a query from \(q \in \alpha_i\) with the corresponding query from \(q\) in \(\alpha_i \cup \beta_i\). After that, we report any farthest points from \(q\) in \(N\) outside of \(\alpha_i\) based on the position of \(q\) along  \(\alpha_i\): When \(q\) lies between \(a_i\) and \(\hat b_i\) (left case), we report the farthest points from \(q\) in \(N \setminus \alpha_i\) with a query from \(a_i\) in \(N\). When \(q\) lies between \(\hat a_i\) and \(b_i\) (right case), we report the farthest points from \(q\) in \(N \setminus \alpha_i\) with a query from \(b_i\) in \(N\). When \(q\) lies between \(\hat b_i\) and \(\hat a_i\) (middle case), we attempt a query from the unique point \(q'\) on \(\beta_i\) with \(d(\hat b_i,q) = d(a_i,q)\).  Since \(q'\) preserves the relative position from \(q\) on \(\alpha_i\), the query from \(q'\) yields the farthest points from \(q\) in \(N\) outside of \(\alpha_i\)---except when \(q'\) has only a single farthest point on \(\alpha_i\). This occurs when \(\hat D\) has the low plateau of \(\hat d_i\) at \(q'\). Fortunately, we are able to recover the correct answer, since the upper envelope of \(\hat d_1,\dots, \hat d_{i-1},\hat d_{i+1},\dots, \hat d_\s\)  coincides at \(q'\) with the partial upper envelope \(\hat D'\) from the proof of Lemma\nobreakspace \ref {thm::envelope}. 

Our data structure has \(O(n)\) size and construction time, since each vertex appears only in a constant number of sub-structures each of which have linear construction time. Every farthest-point query takes \(O(k + \log n)\) time, because each query consists of a constant number of \(O(\log n)\)-time farthest-distance queries followed by a constant number of farthest point queries only in those sub-structures that actually contain farthest points from the original query. \qed
\end{proof}


\section{Abacus Networks} \label{sec::abacus}

An \emph{abacus} is a network \(A\) consisting of a parallel-path network \(N\) with arcs attached to its parallel paths, as illustrated in Fig.\nobreakspace \ref {fig::abacus}. Let \(P_1,\dots, P_\p\) be the parallel paths of \(N\) and let \(B_i\) be the \(i\)-th parallel path with attached arcs.  

\begin{figure}[htb]
	\centering 
	\includegraphics[]{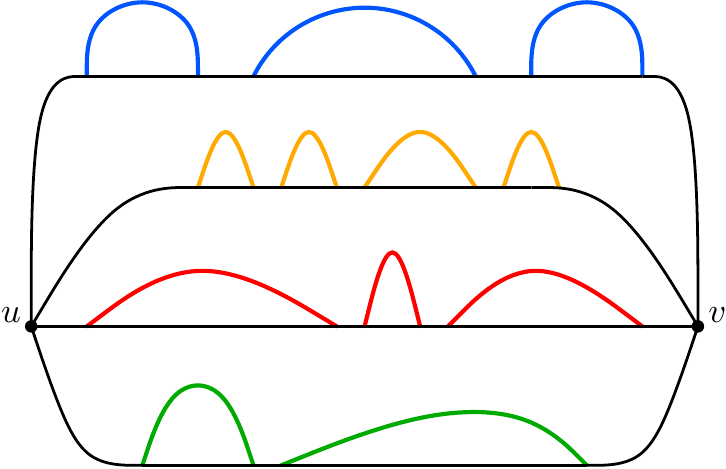}
	\caption{An abacus with the arcs (colored) attached to its parallel-path network (black).\label{fig::abacus}}
\end{figure}

We split farthest-point queries in an abacus into an \emph{inward query} and an \emph{outward query}: an inward query considers farthest points on the bead-chain containing the query point; an outward query considers farthest points on the remaining bead-chains. We first perform the farthest distance version of inward and outward queries before reporting farthest points where appropriate. Figure\nobreakspace \ref {fig::abacus::inwardoutward} illustrates how we treat inward and outward queries in the following. 

For an inward query from \(q\) on \(B_i\), we construct the bead-chain network~\(B_i'\) consisting of \(B_i\) with an additional edge from \(u\) to \(v\) of weight \(d(u,v)\), as illustrated in Fig.\nobreakspace \ref {fig::abacus::inwardoutward::inward}. Since \(B_i'\) preserves distances from \(A\), the farthest points from \(q\in B_i'\) are the farthest points from \(q\) among the points on \(B_i\) in \(A\). 

\begin{figure}
	\subfloat[][An inward query.]{\includegraphics[page=2,scale=0.8]{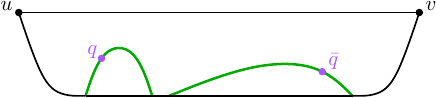}\label{fig::abacus::inwardoutward::inward}} \quad
	\subfloat[][The left case of an outward query.]{\includegraphics[page=3,scale=0.8]{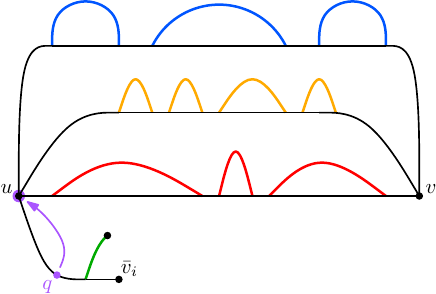}\label{fig::abacus::inwardoutward::left}} 
	
	\subfloat[][The right case of an outward query.]{\includegraphics[page=4,scale=0.8]{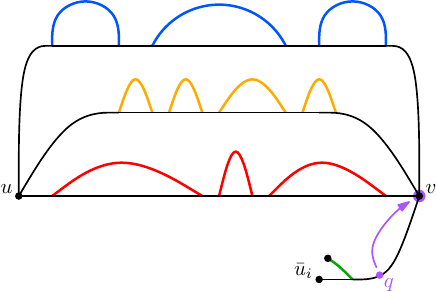}\label{fig::abacus::inwardoutward::right}} \quad
	\subfloat[][An outward query from an arc.]{\includegraphics[page=5,scale=0.8]{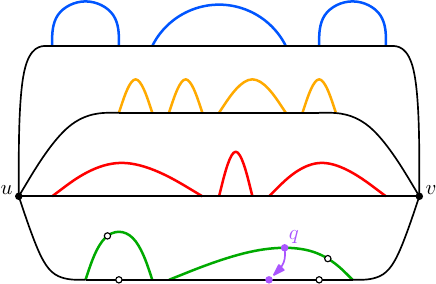}\label{fig::abacus::inwardoutward::arc}} 
	
	\subfloat[][Translating a query to the  virtual edge.]{\includegraphics[page=6,scale=0.8]{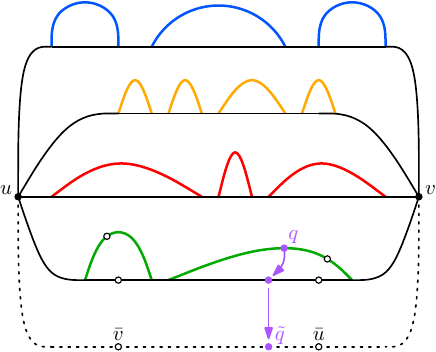}\label{fig::abacus::inwardoutward::virtual}} \quad
	\subfloat[][Collapsing the chains. ]{\includegraphics[page=7,scale=0.8]{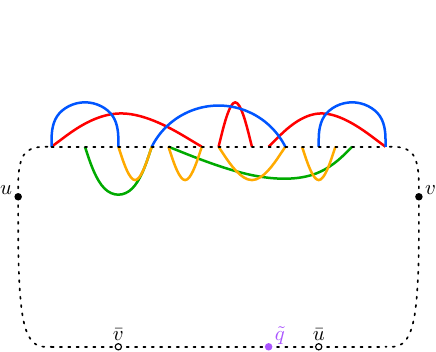}\label{fig::abacus::inwardoutward::collapse}} 
	
	\caption{Inward (a) and outward (b--f) queries for the abacus network from~Fig.\nobreakspace \ref {fig::abacus}. Inward queries are answered in the bead-chain containing the query (a). Outward queries in the side case are answered with queries form the terminals (b,c). Outward queries in the middle case from arcs are translated to queries from the path (d) and then to queries from a virtual edge (e). From the perspective of the virtual edge, we conceptually collapse all bead-chains of the abacus to support virtual queries (f).    \label{fig::abacus::inwardoutward}}
\end{figure}

For outward queries in an abacus, we distinguish the same three cases as for parallel-path networks: we are in the left case when every shortest path tree reaches \(u\) before \(v\), we are in the right case when every shortest path tree reaches \(v\) before \(u\), and we are in the middle case otherwise. Analogously to Lemma\nobreakspace \ref {thm::parallelfirst::threecases}, the left case applies when we are within distance \(d(u, \bar v_i)\) from \(u\) and the right case applies when we are within distance \(d(v,\bar u_i)\) from \(v\). 

For an outward query from \(q \in B_i\) in the left case, \(q\) has the same farthest points as \(u\) outside of \(B_i\). During the construction of the networks \(B_1', \dots, B_\p'\) for inward queries, we determine a list \(L_j\) of the farthest points from \(u\) in \(B_j'\). Similarly to our treatment of the left case for parallel-path networks, we only keep the list achieving the highest farthest distance and the lists achieving the second highest farthest distance. With this preparation, answering the query for~\(q\) amounts to reporting the entries of the appropriate lists \(L_j\) with \(j \ne i\). 

For middle case outward queries, we proceed along the following four steps: First, we translate every outward query from an arc of \(B_i\) to an outward query from the path \(P_i\), i.e., we argue that it suffices to consider outward queries from the parallel paths (Fig.\nobreakspace \ref {fig::abacus::inwardoutward::arc}). Second, we translate outward queries from \(P_i\) to outward queries from a virtual edge \(\tilde e\) connecting the terminals (Fig.\nobreakspace \ref {fig::abacus::inwardoutward::virtual}). Third, we speed up queries from the virtual edge by superimposing the data structures for the bead-chains \(B_1 \cup \tilde e\), \dots, \(B_\p \cup \tilde e\), i.e., by conceptually collapsing the parallel chains (Fig.\nobreakspace \ref {fig::abacus::inwardoutward::collapse}). Finally, we recover the correct answer to the original outward query from the answer obtained with an outward query from the virtual edge. 

\begin{lemma}
Let \(\alpha\) be an arc in an abacus and let \(\beta\) be the other path connecting the endpoints of \(\alpha\). For every point \(q \in \alpha\) in the middle case, there is a point \(q' \in \beta\) such that \(q'\) has the same outward farthest points as \(q\).  
\end{lemma}%
\begin{proof}
We continuously deform \(\alpha\) to \(\beta\) maintaining the relative position of \(q\) to the endpoints of \(\alpha\). The distance from \(q\) to all points in the network decreases at the same rate, hence, the outward farthest points remain the same. \qed
\end{proof}

We introduce a virtual edge \(\tilde{e}\) from \(u\) to \(v\) of length \(w_\p\), i.e., the length of the longest \(u\)-\(v\)-path \(P_\p\) in the underlying parallel-path network, as illustrated in Fig.\nobreakspace \ref {fig::abacus::inwardoutward::virtual}. Let \(\bar u\) be the farthest point from \(u\) on \(\tilde{e}\) and let \(\bar v\) be the farthest point from \(v\) on \(\tilde{e}\). From Lemma\nobreakspace \ref {thm::parallelwalk}, we know that the sub-edge \(\bar u\bar v\) of \(\tilde e\) has length \(d(u,v)\) and, thus, the same length as the sub-path from \(\bar u_i\) to \(\bar v_i\) on each parallel path \(P_i\). We translate an outward query from \(q \in P_i\) to a query from the unique point \(\tilde q\) on \(\tilde{e}\) such that \(\tilde q\) has the same distance to \(\bar u\) and to \(\bar v\) as \(q\) to \(\bar u_i\) and to \(\bar v_i\). 

\begin{lemma}
For \(q \in P_i\) in the middle case, the farthest points from \(q\) in \(P_i \cup B_j\) are the farthest points from \(\tilde q\) in \(\tilde{e} \cup B_j\) for every \(j \ne i\). 
\end{lemma}%
\begin{proof} We continuously elongate \(P_i\) to \(\tilde{e}\) maintaining the relative position of \(q\) to \(u\) and \(v\) and, thus, to \(\bar v_i\) and \(\bar u_i\). At the end of this process \(q\) coincides with \(\tilde q\). The distance from \(q\) to all points outside of \(B_i\) increases uniformly. Hence, the outward farthest points remain the same throughout the deformation. \qed
\end{proof}

It would be too inefficient to inspect each bead-chain network \(B_j \cup \tilde{e}\) with \(j\ne i\) to answer an outward query from \(q \in P_i\). Instead, we first determine the upper envelopes of the farthest-arc distances \(\hat D_1,\dots, \hat D_\p\) along \(\tilde{e}\) in each \(B_1 \cup \tilde{e},\dots,B_\p \cup \tilde{e}\) and then compute their upper envelope \(U_1\) as well as their second level \(U_2\), i.e., the upper envelope of what remains when we remove the segments of the upper envelope. Computing the upper envelope and the second level takes \(O(n \log \p)\) time, e.g., using plane sweep. Using fractional cascading, we support constant time jumps between corresponding segments of \(U_1\) and \(U_2\). The resulting structure occupies \(O(n)\) space, since each of the \(O(n)\) arcs along any bead-chain contributes at most four bending points to \(U_1\) and \(U_2\).

We answer an outward query from \(q \in P_i\) in the middle case by translating \(q\) to \(\tilde q\). When the segment defining \(U_1\) at \(\tilde{q}\) is from some arc \(\alpha\) of \(B_j\) with \(j\ne i\), then \(\alpha\) contains an outward farthest point from \(q\). When the segment defining \(U_1\) at \(\tilde{q}\) corresponds to an arc of \(B_i\), then we jump down to \(U_2\), which leads us to an arc containing an outward farthest point from \(q\). We report the remaining arcs with outward farthest points by walking \(\tilde{q}\) along \(U_1\) and \(U_2\). In order to skip long sequences of segments from \(B_i\), we introduce pointers along \(U_1\) to the next segment from another chain in either direction. Answering outward queries in the middle case takes \(O(k  + \log n)\) time  after \(O(n \log p)\) construction time.

\begin{theorem}
Let \(N\) be an abacus with \(n\) vertices and \(\p\) chains. There is a data structure of size \(O(n)\) with \(O(n \log \p)\) construction time  supporting  farthest-point queries on \(N\) in \(O(k + \log n)\) time, where \(k\) is the number of farthest points. 
\end{theorem}

\section{Two-Terminal Series-Parallel Networks} \label{sec::twoterminal}

 Consider a two-terminal series-parallel network \(N\). By undoing all possible series operations and all possible parallel operations in alternating rounds, we reduce~\(N\) to an edge connecting its terminals and decompose \(N\) into paths that reflect its creation history. The colors in Fig.\nobreakspace \ref {fig::seriesparallel::creation} illustrate this decomposition.

\begin{figure}
	\centering \vspace{-1\baselineskip}
	\includegraphics[scale=0.9, page=2]{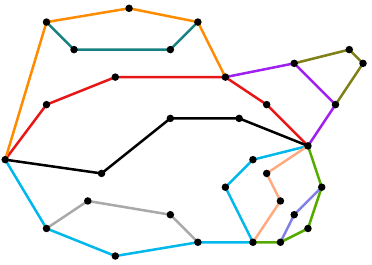}
	\caption{A two-terminal series-parallel network with colors indicating the parallel operations in a potential creation history: starting with a single red edge, we create a parallel yellow and a parallel blue edge. Then we subdivide the blue edge using series operations until we create a parallel purple edge for one of the blue edges and so forth.\label{fig::seriesparallel::creation}}  \vspace{-1\baselineskip}
\end{figure} 

\begin{lemma} \label{thm::reconstruction}
Let \(N\) be a series-parallel network with parallelism \(\p\) and serialism~\(\s\). Identifying the terminals of \(N\) and reconstructing its creation takes \(O(\s + \p)\) time.  
\end{lemma}
\begin{proof} Series-parallel networks are planar~\cite{duffin1965topology}. We maintain a series-parallel network \(N\) together with its dual \(N^*\) throughout the following reduction process illustrated in Fig.\nobreakspace \ref {fig::thm::reconstruction}. We keep two arrays \(d\) and \(d^*\) with \(n\) entries each where \(d[i]\) stores a list of the vertices of \(N\) with degree \(i\) and \(d^*[i]\) stores a list of the vertices of \(N^*\) (faces of \(N\)) with degree \(i\). Each vertex of \(N\) and each face of \(N\)  maintain a pointers to their position in the lists to facilitate constant time deletions. 

We proceed in alternating rounds where we either reverse  as many series operations or as many parallel operations as possible in each round. To reverse series operations we delete all degree two vertices of \(N\). 
\begin{figure}[p]
	\subfloat[][Step 1: reverting series operations.]{\includegraphics[page=3]{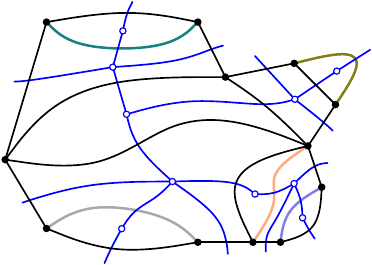}}\quad	
	\subfloat[][Step 2: reverting parallel operations.]{\includegraphics[page=4]{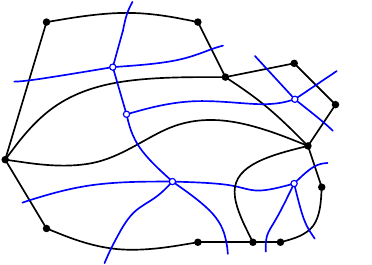}}
	
	\subfloat[][Step 3: reverting series operations.]{\includegraphics[page=5]{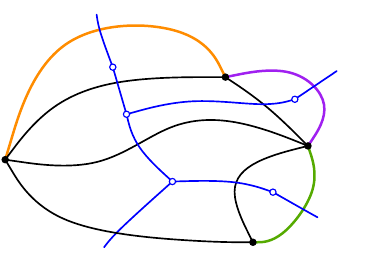}}\quad		
	\subfloat[][Step 4: reverting parallel operations.]{\includegraphics[page=6]{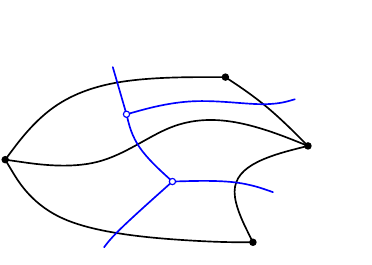}}
	
	\subfloat[][Step 5: reverting series operations.]{\includegraphics[page=7]{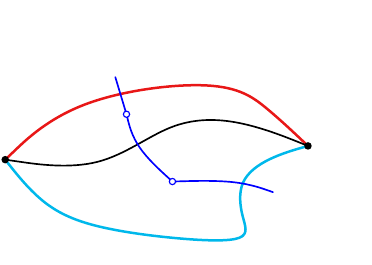}}\quad
	\subfloat[][Step 6: reverting parallel operations.]{\includegraphics[page=8]{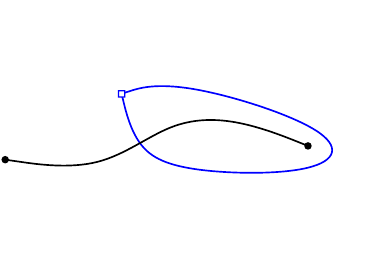}} 
	
	\caption{The reduction process for the two-terminal series-parallel network from Fig.\nobreakspace \ref {fig::seriesparallel::creation}. In each step, we contract all degree two vertices in the primal network to revert series operations or we contract all degree two vertices in the dual (blue) to revert parallel operations. When reverting parallel operations, the colors indicate how we backtrack the creation history in Fig.\nobreakspace \ref {fig::seriesparallel::creation}. For the sake of clarity, the dual vertex for the outer face (blue square) is only shown in the last step where no further reduction is possible.  \label{fig::thm::reconstruction}}
\end{figure}

Each deletion changes the degree of two faces of \(N\) that were incident to the removed vertex so we move these faces to their new positions in \(d^*\). To reverse parallel operations we proceed in the exact same fashion by removing all degree two vertices of \(N^*\) and updating \(d\) accordingly. We recover the creation history of \(N\) by keeping track of when parallel edges were removed during the reversal of a parallel operation. This reduces \(N\) to the edge connecting its terminals in \(O(\s + \p)\) steps, since reversing each of the \(\s + \p\) operations takes constant time.  \qed
\end{proof}

Once we know the terminals \(u\) and \(v\) of \(N\), we compute the shortest path distances from \(u\) and from \(v\) in  \(O(n \log \p)\) time.\footnote{The priority queue in Dijkstra's algorithm manages never more than \(\p\) entries.} Consulting the creation history, we determine a maximal parallel-path sub-network \(P\) of \(N\) with terminals \(u\) and \(v\). As illustrated in Figs.\nobreakspace \ref {fig::nesting} and\nobreakspace  \ref {fig::nestingtree}, every bi-connected component \(X\) of \(N\) that is attached to some path of \(P\) between vertices \(a\) and \(b\) is again a two-terminal series-parallel network with terminals \(a\) and \(b\). We recurse on these bi-connected components. When this recursion returns, we know a longest \(a\)-\(b\)-path in \(X\) and attach an arc from \(a\) to \(b\) of this length to \(P\). The resulting network is an abacus \(A\). The abaci created during the recursion form a tree \(\mathcal{T}\) with root~\(A\). Alongside with this decomposition we also create our data structures for the nested abaci. 

The size of the resulting data structure remains \(O(n)\), since the data structure for each nested abaci consumes space linear in the number of its vertices and each vertex in any nested abaci can be charged to one of the series or parallel operations required to generate the original network.

\begin{figure}[htpb]
	\centering
	\subfloat[][A two-terminal series-parallel network.]{\includegraphics[page=]{nesting}}%
	
	\subfloat[][]{\includegraphics[scale=0.67, page=2]{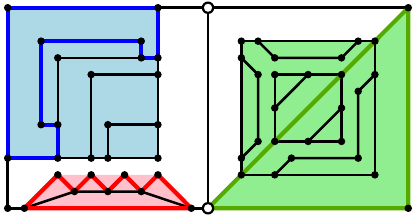}\quad\includegraphics[scale=0.67, page=3]{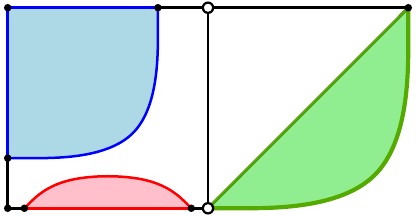}}	
	
	\subfloat[][]{\includegraphics[scale=0.67, page=5]{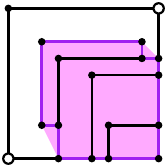}\quad\includegraphics[scale=0.67, page=6]{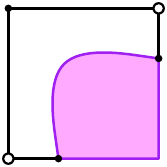}}\quad	
	\subfloat[][]{\includegraphics[scale=0.67, page=4]{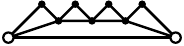}}\quad
	\subfloat[][]{\includegraphics[scale=0.67, page=10]{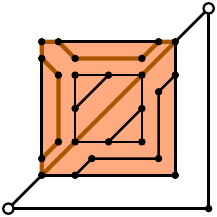}\quad\includegraphics[scale=0.67, page=11]{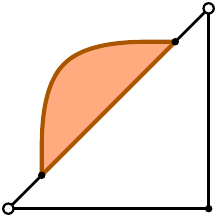}}
	
	\subfloat[][]{	\includegraphics[scale=0.67, page=7]{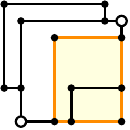}\quad%
					\includegraphics[scale=0.67, page=8]{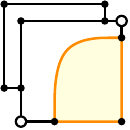}}\quad%
	\subfloat[][]{	\includegraphics[scale=0.67, page=9]{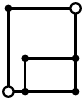}}\hfill
	\subfloat[][]{	\includegraphics[scale=0.67, page=12]{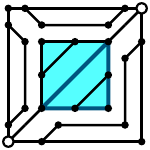}\quad%
					\includegraphics[scale=0.67, page=13]{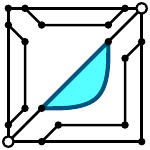}}\quad%
	\subfloat[][]{	\includegraphics[scale=0.67, page=14]{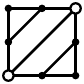}}
	
	\caption{ Decomposing a two-terminal series-parallel network (a) into nested abaci. Each involved network is two-terminal series-parallel and we mark the terminals with empty circles. Replacing the nested structures in each involved network (b--g) yields the abacus network shown to its right. Each colored cycle consists of the shortest path and the longest path connecting the terminals of each nested structure \(X\); their weighted lengths determine the weight of the corresponding arc replacing \(X\). This process terminates when we reach an abacus (d, g, i), i.e., when there are no more nested structures. \label{fig::nesting}} 
\end{figure}

\begin{figure}
	\centering
	\includegraphics{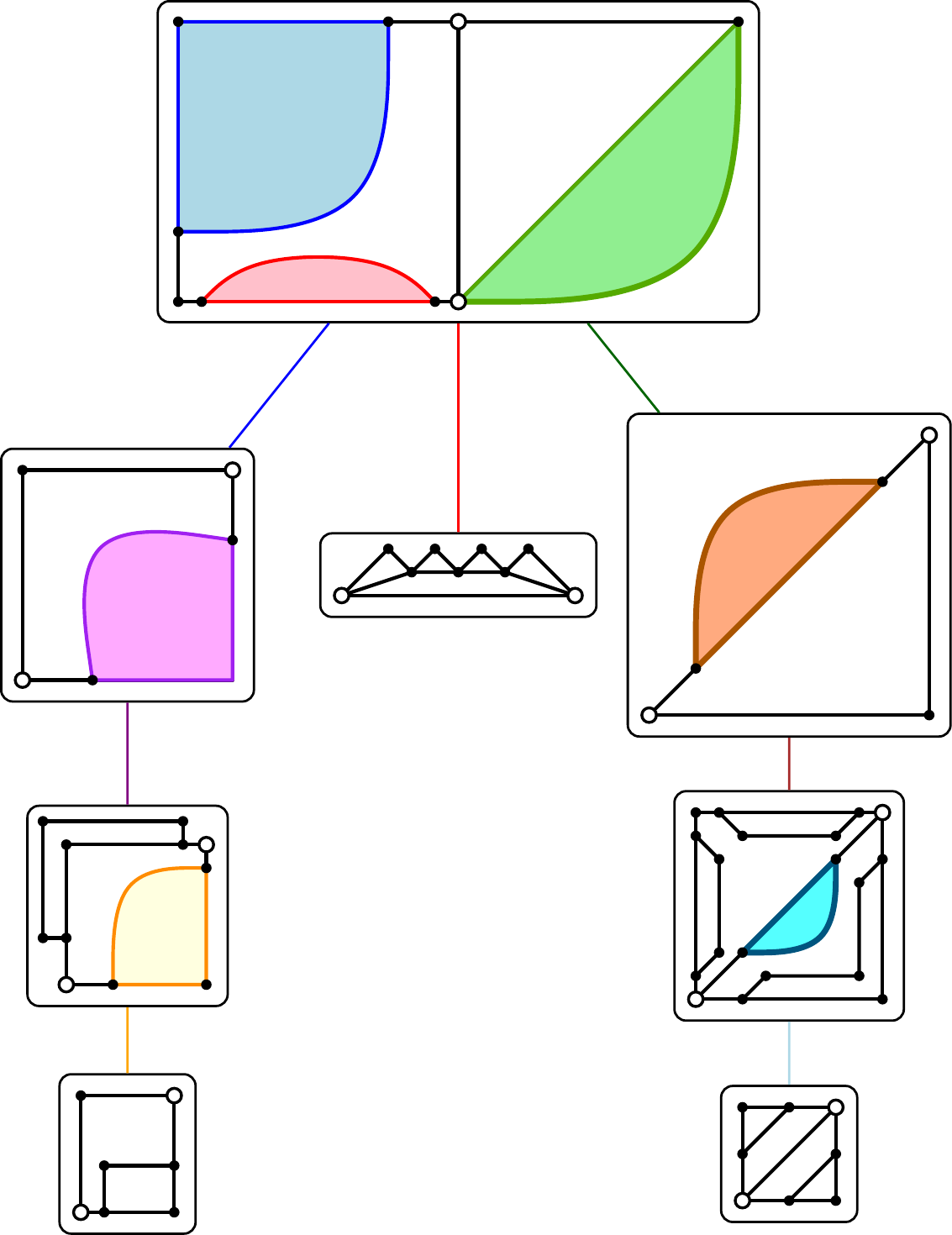}
	\caption{The tree of nested abaci for the two-terminal series-parallel network from Fig.\nobreakspace \ref {fig::nesting}. The inner nodes of this tree correspond to two-terminal networks with nested structures that are indicated with colors; the leaves correspond to abacus networks without nested structures. A query would start at the root abacus and cascade into nested structures when necessary. For instance, when a query at the root abacus yields a farthest point on the blue arc and a farthest point on the red arc, we would perform subsequent queries in the abaci stored in the left and middle child of the root.  \label{fig::nestingtree}}
\end{figure}

We translate a query \(q\) to a query in the abacus~\(A\); queries from a bi-connected component \(X\) attached to \(P\) in \(N\) will be placed on the corresponding arc of \(A\). Whenever the query in \(A\) returns a farthest point on some arc, we cascade the query into the corresponding nested data structure. We add shortcuts to the abacus tree \(\mathcal{T}\) in order to avoid cascading through too many levels of \(\mathcal{T}\) without encountering farthest-points from the original query. This way, answering farthest-point queries in \(N\) takes \(O(k + \log n)\) time in total. 

\begin{theorem}
Let \(N\) be a two-terminal series-parallel network with \(n\) vertices and parallelism \(\p\). There is a data structure of size \(O(n)\) with \(O(n \log \p)\) construction time  that supports  \(O(k + \log n)\)-time farthest-point queries from any point on \(N\), where \(k\) is the number of farthest points. 
\end{theorem}

\section{Conclusion and Future Work} \label{sec::conclusion}

In previous work, we learned how to support farthest-point queries by exploiting the treelike structure of cactus networks~\cite{bose2015optimal}. In this work, we extended the arsenal by techniques for dealing with parallel structures, as well. In future work, we aim to tackle more types of networks such as planar networks, \(k\)-almost trees~\cite{gurevich1984solving}, or generalized series-parallel networks~\cite{korneyenko1994combinatorial}. Moreover, we are also interested in lower bounds on the construction time of data structures supporting efficient farthest-point queries to guide our search for optimal data structures. 

\bibliography{twoterminal.bib}{}
\bibliographystyle{splncs03}

\end{document}